\documentclass{article}

\usepackage{cmbright}
\usepackage{amssymb,amsmath,amsthm}
\usepackage{mathrsfs}
\usepackage{epsfig}
\usepackage{color}
\usepackage{stmaryrd}

\def\B{\mathscr B}
\def\C{\mathbb C}

\def\d{\mathrm d}
\def\F{\mathscr F}
\def\G{\mathcal G}
\def\H{\mathcal H}
\def\M{\mathsf M}
\def\N{\mathbb N}
\def\O{\mathcal O}
\def\P{\mathcal P}
\def\R{\mathbb R}

\def\Hrond{\mathscr H}
\def\Oa{{\mathcal O}_{\rm as}}

\def\ltwo{\mathsf{L}^{\:\!\!2}}
\def\linf{\mathsf{L}^{\:\!\!\infty}}
\def\e{\mathop{\mathrm{e}}\nolimits}
\DeclareMathOperator*{\im}{Im}
\DeclareMathOperator*{\re}{Re}

\DeclareMathOperator*{\slim}{s\;\!-lim\;\!}

\def\Ran{\mathop{\mathsf{Ran}}\nolimits}
\DeclareMathOperator*{\Tr}{Tr}

\newtheorem{Theorem}{Theorem}[section]

\newtheorem{Lemma}[Theorem]{Lemma}

\newtheorem{Corollary}[Theorem]{Corollary}
\newtheorem{Proposition}[Theorem]{Proposition}

\begin{document}


\title{Resolvent expansions and continuity of the scattering matrix\\
at embedded thresholds: the case of quantum waveguides}

\author{S. Richard$^1$~~and R. Tiedra de
Aldecoa$^2$\footnote{Supported by the Chilean Fondecyt Grant 1130168 and by the
Iniciativa Cientifica Milenio ICM RC120002 ``Mathematical Physics'' from the Chilean
Ministry of Economy.}}

\date{\small}
\maketitle \vspace{-1cm}

\begin{quote}
\emph{
\begin{itemize}
\item[$^1$] Graduate school of mathematics, Nagoya University,
Chikusa-ku, Nagoya 464-8602, Japan; On leave of absence from
Universit\'e de Lyon; Universit\'e
Lyon 1; CNRS, UMR5208, Institut Camille Jordan,
43 blvd du 11 novembre 1918, F-69622
Villeurbanne-Cedex, France.
\item[$^2$] Facultad de Matem\'aticas, Pontificia Universidad Cat\'olica de Chile,\\
Av. Vicu\~na Mackenna 4860, Santiago, Chile
\item[] \emph{E-mails:} richard@math.univ-lyon1.fr, rtiedra@mat.puc.cl
\end{itemize}
}
\end{quote}


\begin{abstract}
We present an inversion formula which can be used to obtain resolvent expansions near
embedded thresholds. As an application, we prove for a class of quantum waveguides
the absence of accumulation of eigenvalues and the continuity of the scattering
matrix at all thresholds.
\end{abstract}


\smallskip


\section{Introduction}\label{Intro}
\setcounter{equation}{0}

During the recent years, there has been an increasing interest in resolvent
expansions near thresholds and their various applications. These developments were
partially initiated by the paper of A. Jensen and G. Nenciu \cite{JN01} in which a
general framework for asymptotic expansions is presented and then applied to
potential scattering in dimension $1$ and $2$. The key point of that paper is an
inversion formula which provides an efficient iterative method for inverting a family
of operators $A(z)$ as $z\to0$ even if $\ker\big(A(0)\big)\ne\{0\}$. Corrections or
improvements of this inversion formula can be found in \cite[Lemma 4]{ES04},
\cite[Prop.~3.2]{IJ13} and \cite[Prop.~1]{JN04}. However, in all these papers either
it is assumed that $A(0)$ is self-adjoint, or the construction relies on a Riesz
projection which is not always convenient to deal with. These features are harmless
in these works, since the threshold considered always lies at the endpoints of the
spectrum of the underlying operator. However, once dealing with embedded thresholds,
these features turn out to be critical (see the comment at the end of Section
\ref{seco}).

Our aim in the present paper is thus twofold. On the one hand, we revisit the
mentioned inversion formula, and on the other hand we show how its revised version
can be used for proving the continuity of a scattering matrix at embedded thresholds.
The abstract part of our results is presented in Section \ref{Sec_Inv}, and
consists first in a reformulation of the inversion formula which does not require that
$A(0)$ is self-adjoint or that the projection is a Riesz projection (see Proposition
\ref{propourversion}). We then discuss two natural choices for the projection\;\!:
either the Riesz projection defined in terms of the resolvent of $A(0)$ if $0$ is an
isolated point in the spectrum of $A(0)$, or the orthogonal projection on
$\ker\big(A(0)\big)$ if $A(0)$ has a non-negative imaginary part. If both conditions
hold, we also discuss the relations between these two projections, and provide
sufficient conditions for their equality. This situation often takes place in
applications even without the assumption that $A(0)$ is self-adjoint (see Corollary
\ref{Corol_So}).

In the second part of the paper (Section \ref{secwave}), we present an application
of our abstract results to scattering theory for quantum waveguides.
Quantum waveguides provide a particularly good model of study since their Hamiltonians
possess an infinite number of embedded thresholds (with a change of multiplicity at each threshold)
but give rise to a simple scattering theory taking place in a one-Hilbert space setting.
We refer to \cite{Tie06} for basic results and earlier references on the spectral
and scattering theory for quantum waveguides.

For a straight quantum waveguide with a compactly supported potential $V$, we derive
an asymptotic expansion of the resolvent in a neighbourhood of each embedded
threshold. More precisely, if the potential is written as $V=vuv$ with $v$
non-negative and $u$ unitary and self-adjoint, and if $H_0$ is the Dirichlet
Laplacian for the waveguide, then we give an expansion of the operator
$\big(u+v(H_0-z)^{-1}v\big)^{-1}$ as $z$ converges to any threshold $z_0$ (see
Proposition \ref{Prop_Asymp}). Note that the operator $v(H_0-z_0)^{-1}v$ (once
properly defined) has a non-trivial imaginary part. This fact automatically prevents
the use of any approach assuming the self-adjointness of $A(0)$, as mentioned above.

We then deduce two consequences of this asymptotic expansion. First, we prove in
Corollary \ref{noaccp} that the possible point spectrum of the operator $H:=H_0+V$
does not accumulate at thresholds. Since the thresholds are the only possible
accumulation points for such a model, we thus rule out this possibility.
Second, we characterize for all scattering channels corresponding to the transverse modes of
the waveguide the behavior of the scattering matrix for the pair $\{H_0,H\}$ at embedded
thresholds. More precisely, we show that the scattering matrix is continuous at the
thresholds if the channels we consider are already open, and that the scattering
matrix has a limit from the right at the thresholds if a channel precisely opens at
these thresholds (see Proposition \ref{propcon} for a precise formulation of this
result). Up to our knowledge, these types of results are completely new since the
analysis of the behavior of a scattering matrix at embedded thresholds has apparently
never been performed. We also show the continuity of the scattering matrix at
embedded eigenvalues which are not located at thresholds. But in this case, similar
results were already known for other models, see for example \cite[Prop.~10]{IR12} or
\cite[Prop.~6.7.11]{Y} (see also \cite{GJY04} where propagation estimates at embedded
thresholds are obtained for a Schr\"odinger operator with time periodic potential).

As a final comment, we stress that we fully describe all possible behaviors at
thresholds since we do not assume any condition on the absence of bound states or
resonances at thresholds. Based on the expressions obtained in this paper, a
Levinson's type theorem for quantum waveguides could certainly be derived, and
deserves further investigations.\\

\noindent
{\bf Acknowledgements.} The authors thank A. Jensen for useful discussions.

\section{Inversion formula}\label{Sec_Inv}
\setcounter{equation}{0}

In this section, we adapt the inversion formula \cite[Prop.~1]{JN04} to the case of
an arbitrary projection, and then discuss two possible choices for
this projection. The symbol $\H$ stands for an arbitrary Hilbert space with norm
$\|\cdot\|$ and scalar product $\langle\;\!\cdot\;\!,\;\!\cdot\;\!\rangle$, and
$\B(\H)$ denotes the algebra of bounded operators on $\H$ with norm also denoted by
$\|\cdot\|$.

\begin{Proposition}\label{propourversion}
Let $O\subset\C$ be a subset with $0$ as an accumulation point. For each $z\in O$,
let $A(z)\in\B(\H)$ satisfy
$$
A(z)=A_0+zA_1(z),
$$
with $A_0\in\B(\H)$ and $\|A_1(z)\|$ uniformly bounded as $z\to0$. Let also
$S\in\B(\H)$ be a projection such that\;\!:
\begin{enumerate}
\item[(i)] $A_0 + S$ is invertible with bounded inverse,
\item[(ii)] $S(A_0+S)^{-1}S = S$.
\end{enumerate}
Then, for $|z|>0$ small enough the operator $B(z): S\H \to S\H$ defined by
\begin{equation}\label{eqB(z)}
B(z)
:=\frac1z\left(S-S\big(A(z)+S\big)^{-1}S\right)
\equiv S(A_0+S)^{-1}\Bigg(\sum_{j\ge0}(-z)^j\big(A_1(z)(A_0+S)^{-1}\big)^{j+1}\Bigg)S
\end{equation}
is uniformly bounded as $z\to0$. Also, $A(z)$ is invertible in $\H$ with bounded
inverse if and only if $B(z)$ is invertible in $S\H$ with bounded inverse, and in
this case one has
$$
A(z)^{-1}
=\big(A(z)+S\big)^{-1}+\frac1z\big(A(z)+S\big)^{-1}SB(z)^{-1}S\big(A(z)+S\big)^{-1}.
$$
\end{Proposition}

\begin{proof}
For $z\in O$ with $|z|>0$ small enough, one has the equality
\begin{align*}
B(z)
&=\frac1z\big(S-S(A_0+S)^{-1}S\big)
+S(A_0+S)^{-1}\Bigg(\sum_{j\ge0}(-z)^j\big(A_1(z)(A_0+S)^{-1}\big)^{j+1}\Bigg)S.
\end{align*}
So, the condition (ii) implies the second equality in \eqref{eqB(z)}.
The second part of the claim is a direct application of the inversion formula
\cite[Lemma~2.1]{JN01}.
\end{proof}

The choice of the projection $S$ plays an important role in the previous proposition.
For example, if $0$ is an isolated point in the spectrum $\sigma(A_0)$ of $A_0$, a
natural candidate for $S$ is the Riesz projection associated with this value, which
is the choice made in \cite{ES04,JN01,JN04}. Another natural candidate is the
orthogonal projection on the kernel of $A_0$. However, for both choices additional
conditions are necessary in order to verify conditions (i) and (ii). Below, we first
discuss the case of the Riesz projection and then the case of the orthogonal
projection.

\subsection{Riesz projection}

In this section, we assume that $0$ is an isolated point in $\sigma(A_0)$ and write
$S_r$ for the corresponding Riesz projection. In that case, $A_0S_r=S_rA_0=S_rA_0S_r$
and $A_0+S_r$ is invertible with bounded inverse (see \cite[Chap.~III.6.4]{Kato}).
The condition (ii) above, namely $S_r(A_0+S_r)^{-1}S_r=S_r$, is more complicated to
check. However, if one assumes that $A_0S_r=0$, or the stronger condition that $A_0$
is self-adjoint, then the equalities $S_r(A_0+S_r)^{-1}=S_r=(A_0+S_r)^{-1}S_r$ hold,
and thus condition (ii) is satisfied (note that in that case a small simplification
takes place on the r.h.s. of \eqref{eqB(z)}). However, the condition $A_0S_r=0$ does
not always hold since $A_0S_r$ is in general only quasi-nilpotent
\cite[Sec.~III.6.5]{Kato}. Fortunately, the condition $A_0S_r=0$ holds if $A_0$ has a
particular form, as shown in the following lemma (which is an extension of
\cite[Prop.~2]{JN04}).

\begin{Lemma}\label{lemme_Riesz}
Assume that $A_0=X+i\;\!Y$, with $X,Y$ bounded self-adjoint operators and $Y\ge0$, and
suppose that $0$ is an isolated point in $\sigma(A_0)$. Let $S_r$ be the
corresponding Riesz projection, and assume that $S_rA_0S_r$ is a trace-class
operator. Then, $A_0S_r=S_rA_0=0$.
\end{Lemma}

Note that the trace-class condition is satisfied if, for instance, $S_r\H$ is
finite-dimensional.

\begin{proof}
Since $S_r$ is a projection which commutes with $A_0$, one has
$A_0S_r=S_rA_0=S_rA_0S_r$. Therefore, if $J$ is the operator in $S_r\H$ given by
$J:=S_rA_0S_r$, then
$$
\im\big\langle S_r\varphi,JS_r\varphi\big\rangle
=\im\big\langle S_r\varphi,S_rA_0S_rS_r\varphi\big\rangle
=\im\big\langle S_r\varphi,A_0S_r\varphi\big\rangle
\ge0\quad\hbox{for all }\varphi\in\H,
$$
or equivalently $\im(J)\ge0$ in $S_r\H$. Since $J$ is quasi-nilpotent
\cite[Eq.~(III.6.28)]{Kato} and trace-class, and since quasi-nilpotent trace-class
operators have trace $0$ \cite[p.~32]{Sim05}, it follows that
$$
0=\Tr(J)=\Tr\big(\re(J)\big)+i\;\!\Tr\big(\im(J)\big).
$$
This equality together with the inequality $\im(J)\ge0$ imply that $\im(J) =0$. Thus,
$J$ is self-adjoint and quasi-nilpotent, which means that $J=0$.
\end{proof}

We now list a series of consequences of the previous result.

\begin{Corollary}\label{Corol_Sr}
Suppose that the assumptions of Lemma \ref{lemme_Riesz} are satisfied, then the
conditions (i) and (ii) of Proposition \ref{propourversion} are verified for $S=S_r$.
\end{Corollary}

\begin{Corollary}\label{Cor_image}
Suppose that the assumptions of Lemma \ref{lemme_Riesz} are satisfied, then
$S_r\H=\ker(A_0)$.
\end{Corollary}

\begin{proof}
The inclusion $S_r\H\subset\ker(A_0)$ follows from the equality $A_0S_r =0$, and the
inclusion $S_r\H\supset\ker(A_0)$ is standard.
\end{proof}

We finally present a simple result which holds under the assumptions of Lemma
\ref{lemme_Riesz}, but can be proved in a slightly more general context. The norms
and scalar products of the different Hilbert spaces are written with the same
symbols.

\begin{Lemma}\label{Cor_magique}
Let $\G$ be an auxiliary Hilbert space, take $Z_n\in\B(\H,\G)$, and assume that the
sum $\sum_nZ_n^*Z_n$ is weakly convergent. Let also $A_0=X+i\sum_nZ_n^*Z_n$, with $X$
a bounded self-adjoint operator in $\H$, and suppose that $S$ is a projection
satisfying $A_0S=0$ and $SA_0=0$. Then, $Z_mS=0$ and $S Z_m^*=0$ for each $m$.
\end{Lemma}

\begin{proof}
Let $\varphi\in\H$. Then, the first identity follows from the equalities
$$
\textstyle
\big\|Z_mS\varphi\big\|^2
\le\big\langle S\varphi,\big(\sum_n Z_n^*Z_n\big)S\varphi\big\rangle
=\im\big\langle S\varphi,\big(X+i\sum_nZ_n^*Z_n\big)S\varphi\big\rangle
=\im\big\langle S\varphi,A_0S\varphi\big\rangle
=0,
$$
and the second identity follows from the equalities
$$
\textstyle
\big\|Z_mS^*\varphi\big\|^2
\le\big\langle S^*\varphi,\big(\sum_n Z_n^*Z_n\big)S^*\varphi\big\rangle
=-\im\big\langle S^*\varphi,\big(X-i\sum_nZ_n^*Z_n\big)S^*\varphi\big\rangle
=-\im\big\langle S^*\varphi,A_0^*S^*\varphi\big\rangle
=0.
$$
\end{proof}

\subsection{Orthogonal projection on the kernel}\label{seco}

In this section, we assume from the beginning that $A_0=X+i\;\!Y$, with $X,Y$ bounded
self-adjoint operators and $Y\ge0$. In that case, one has
$\ker(A_0)=\ker(X)\cap\ker(Y)=\ker(A_0^*)$. Also, if $S_o$ denotes the orthogonal
projection on $\ker(A_0)$, the relations $XS_o=0=S_oX$, $YS_o=0=S_oY$ and
$A_0S_o=0=S_oA_0$ hold. Thus, if one shows that $A_0+S_o$ is invertible with bounded
inverse, then the conditions (i) and (ii) of Proposition \ref{propourversion} would
follow. So, we concentrate in the sequel on this invertibility condition.

Since $A_0$ is reduced by the orthogonal decomposition $\H=S_o\H\oplus(1-S_o)\H$ and
since $A_0$ is trivial in the subspace $S_o\H$, the operator $A_0+S_o$ is invertible
with bounded inverse if the restriction of $A_0$ to $S_o^\bot\H:=(1-S_o)\H$ is
invertible with bounded inverse. However, since $A_0|_{S_o^\bot \H}$ has an inverse
on $\Ran\big(A_0|_{S_o^\bot\H}\big)=\Ran(A_0)$, and since $\Ran (A_0)$ is dense in
$S_r^\bot\H$ (because
$\overline{\Ran(A_0)}=\ker(A_0^*)^\bot=\ker(A_0)^\bot=S_r^\bot \H$), the only
remaining question concerns the boundedness of the inverse $A_0^{-1}$ on $\Ran(A_0)$.

In the following two lemmas, we exhibit conditions under which this question can be
answered affirmatively.

\begin{Lemma}\label{Lemma_trace}
Assume that $A_0=X+i\;\!Y$, with $X,Y$ bounded self-adjoint operators and $Y\ge0$, and
suppose that $0$ is an isolated point in $\sigma(A_0)$. Let $S_r$ denote the
corresponding Riesz projection, and assume that $S_rA_0S_r$ is a trace-class
operator. Then, $A_0$ is invertible in $\ker(A_0)^\bot$ with bounded inverse if and
only if $S_r$ is an orthogonal projection.
\end{Lemma}

Before giving the proof, we recall that if $S_r$ is an orthogonal projection, then it
automatically follows from Corollary \ref{Cor_image} that $S_r=S_o$.

\begin{proof}
Sufficient condition\;\!: Assume that $S_r$ is an orthogonal projection (and thus
equal to $S_o$). Since $A_0$ is invertible in $S_r^\bot\H$ with bounded inverse by
\cite[Thm.~III.6.17]{Kato}, one infers that $A_0$ is invertible in
$S_o^\bot\H=\ker(A_0)^\bot$ with bounded inverse.

Necessary condition\;\!: Suppose by absurd that $S_r$ is not an orthogonal
projection, or more precisely that $S_r^\bot\H\ne S_o^\bot\H$ (since we already know
that $S_r\H=\ker(A_0)=S_o\H$ by Corollary \ref{Cor_image}). Then, if there exists
$\varphi\in S_r^\bot\H\setminus\{0\}$ with $\varphi\not\in S_o^\bot\H$, one has
$S_o\varphi\ne0$ and $S_o^\bot\varphi\ne0$, and for any $z\in\C\setminus\{0\}$ with
$|z|$ small enough
$$
(A_0-z)^{-1}\varphi=(A_0-z)^{-1}S_o\varphi+(A_0-z)^{-1}S_o^\bot\varphi.
$$
Now, we know from \cite[Thm.~III.6.17]{Kato} that the l.h.s. has a limit in $\H$ as
$z\to0$. But since $S_o\varphi\in\ker(A_0)$, the first term on the r.h.s. does not
have a limit as $z \to 0$. Therefore, the second term on the r.h.s. neither has a
limit as $z \to 0$, and thus the operator $A_0$ is not invertible in
$S_o^\bot\H=\ker(A_0)^\bot$.

On the other hand, if there exists $\varphi\in S_o^\bot\H\setminus\{0\}$ with
$\varphi\notin S_r^\bot\H$, one has $S_r\varphi\ne0$ and $S_r^\bot\varphi\ne0$, and
for any $z\in\C\setminus\{0\}$ with $|z|$ small enough
$$
(A_0-z)^{-1}\varphi=(A_0-z)^{-1}S_r\varphi+(A_0-z)^{-1}S_r^\bot\varphi.
$$
In this case, the second term on the r.h.s. does have a limit in $\H$ as $z\to0$, but
the first term on the r.h.s. does not. Therefore, the l.h.s. does not have a limit in
$\H$ as $z \to 0$, and thus the operator $A_0$ is not invertible in
$S_o^\bot\H=\ker(A_0)^\bot$.

Summing up, if $S_r^\bot\H\ne S_o^\bot\H$, then $A_0$ is not invertible in
$S_o^\bot\H=\ker(A_0)^\bot$, which concludes the proof of the claim.
\end{proof}

\begin{Lemma}\label{lemourcase}
Assume that $A_0=X+i\;\!Y$, with $X,Y$ bounded self-adjoint operators and $Y\ge0$.
Suppose also that $A_0=U+K$ with $U$ unitary and $K$ compact, or that $A_0$ is a
finite-rank operator. Then, $A_0$ is invertible in
$\ker(A_0)^\bot$ with bounded inverse.
\end{Lemma}

\begin{proof}
Recall that $\Ran\big(A_0|_{\ker(A_0)^\bot}\big)\equiv\Ran(A_0)$ is dense in
$S_r^\bot\H$. So, the boundedness of the inverse of $A_0$ in $\ker(A_0)^\bot$ follows
from the closed graph theorem \cite[Thm.~III.5.20]{Kato} if $\Ran(A_0)$ is closed.
But, this is verified under both conditions. Under the first condition, one has
$A_0=U+K=(1+KU^{-1})U$ with $KU^{-1}$ is compact. So, $(1+KU^{-1})$ is Fredholm, and
the image of $U\H=\H$ by $(1+KU^{-1})$ is closed \cite[Thm.~4.3.4]{Davies}. And under
the second condition, $\Ran(A_0)$ is finite-dimensional and thus closed.
\end{proof}

Under the assumptions of Lemma \ref{lemourcase}, the value $0$ is an isolated point
in $\sigma(A_0)$. Thus, the Riesz projection $S_r$ is well defined, and one obtains
the following by combining the two previous lemmas\;\!:

\begin{Corollary}\label{Corol_So}
Suppose that the assumptions of Lemma \ref{lemourcase} are satisfied. Then,
$S_r=S_o$, and the conditions (i) and (ii) of Proposition \ref{propourversion} are
verified for $S=S_r=S_o$.
\end{Corollary}

\begin{proof}
We know from Lemma \ref{lemourcase} that $A_0$ is invertible in $\ker(A_0)^\bot$ with
bounded inverse. Thus, it follows from Lemma \ref{Lemma_trace} that $S_r=S_o$ and
that the conditions (i) and (ii) of Proposition \ref{propourversion} are verified for
$S=S_r=S_o$ if $S_rA_0S_r$ is a trace-class operator. But, the operator $S_rA_0S_r$
is clearly trace-class if $A_0$ is a finite-rank operator. On the other hand, if
$A_0=U+K$ with $U$ unitary and $K$ compact, then the isolated eigenvalue $0$ is of
finite multiplicity, $S_r\H$ is finite-dimensional \cite[Remark III.6.23]{Kato}, and
$S_rA_0S_r$ is also trace-class.
\end{proof}

We close this section with a comment on the usefulness of Corollary \ref{Corol_So} for
the iterative procedure of the next section. If we use a Riesz projection $S_r$
without knowing that it is orthogonal, this is harmless at the first step of the
iteration (as illustrated in \cite{JN04}), but this becomes more and more annoying at
each step of the iteration. Indeed, conjugation by Riesz projections does not preserve
positivity, and thus any argument based on positivity can hardly be invoked.
Therefore, Corollary \ref{Corol_So} leads to various simplifications in the iterative
procedure since it provides conditions guaranteeing that $S_r$ is orthogonal.

\section{Quantum waveguides}\label{secwave}
\setcounter{equation}{0}

We introduce in this section the model of quantum waveguide we use and recall some of
its basics properties. Much of the material is borrowed from \cite{Tie06} to which we
refer for further information.

We consider a bounded open connected set $\Sigma\subset\R^{d-1}$ with $d\ge 2$, and
let $-\Delta^\Sigma_{\rm D}$ be the Dirichlet Laplacian on $\Sigma$ acting in
$\ltwo(\Sigma)$. This operator has a purely discrete spectrum
$\tau:=\{\lambda_n\}_{n\ge1}$ consisting in eigenvalues
$\lambda_1\le\lambda_2\le\cdots$ repeated according to multiplicity. The
corresponding set of eigenvectors is denoted by $\{f_n\}_{n\ge1}$ and the
corresponding set of one-dimensional orthogonal projections is denoted by
$\{\P_n\}_{n\ge1}$. Sometimes, we omit for simplicity to stress that $n\ge1$.

We consider also the straight waveguide $\Omega:=\Sigma\times\R$ with coordinates
$(\omega,x)$, the Hilbert space $\H:=\ltwo(\Omega)$, and the Dirichlet Laplacian
$H_0:=-\Delta^\Omega_{\rm D}$ on $\Omega$ acting in $\H$. This operator decomposes as
$H_0=-\Delta^\Sigma_{\rm D}\otimes1+1\otimes P^2$ in
$\H\simeq\ltwo(\Sigma)\otimes\ltwo(\R)$, with $P:=-i\hspace{1pt}\partial_x$ the usual
self-adjoint operator of differentiation in $\ltwo(\R)$. The spectrum $\sigma(H_0)$ of
$H_0$ is purely absolutely continuous with $\sigma(H_0)=[\lambda_1,\infty)$, and each
value $\lambda\in\tau$ is a threshold in $\sigma(H_0)$ with a change of multiplicity.
Moreover, for $z\in\C\setminus\R$, the resolvents $R^0(z):=(P^2-z)^{-1}$ and
$R_0(z):=(H_0-z)^{-1}$ satisfy the relation
\begin{equation}\label{factor_res}
R_0(z)=\sum_n\P_n\otimes R^0(z-\lambda_n),\quad z\in\C\setminus\R,
\end{equation}
and the resolvent $R^0(z)$ has integral kernel
\begin{equation}\label{eq_noyau_1}
R^0(z)(x,x')=\frac i{2\sqrt z}\e^{i\sqrt z\;\!|x-x'|}\;\!,
\quad z\in\C\setminus\R,~x,x'\in\R,
\end{equation}
with the convention that $\im(\sqrt z)>0$ for $z\in\C\setminus[0,\infty)$.

In the following lemma, we recall some weighted estimates for $R^0(z)$ which
complement the asymptotic expansion given in \cite[Lemma~5.1]{JN01}. We use the
notations $\C_+:=\{z\in\C\mid\im(z)>0\}$ and $\langle x\rangle:=(1+x^2)^{1/2}$, and we
let $Q$ denote the self-adjoint multiplication operator by the variable in
$\ltwo(\R)$.

\begin{Lemma}\label{lemsauve}
Fix $\varepsilon>0$, take $\lambda\in\R\setminus(-\varepsilon,\varepsilon)$ and let
$\zeta\in\overline{\C_+}$ with $|\zeta|<\varepsilon/2$.
\begin{enumerate}
\item[(a)] If $s>1/2$, then the limit
$$
\langle Q\rangle^{-s}R^0(\lambda+\zeta)\langle Q\rangle^{-s}
:=\lim_{\zeta'\to\zeta,\,\zeta'\in\C_+}
\langle Q\rangle^{-s}R^0(\lambda+\zeta')\langle Q\rangle^{-s}
$$
exists in $\B\big(\ltwo(\R)\big)$ and is independent of the sequence $\zeta'\to\zeta$.
Moreover, the limit is a Hilbert-Schmidt operator with Hilbert-Schmidt norm
$$
\big\|\langle Q\rangle^{-s}R^0(\lambda+\zeta)\langle Q\rangle^{-s}\big\|_{\rm HS}
\le{\rm Const.}\;\!|\lambda|^{-1/2}.
$$
\item[(b)] If $s >3/2$, then
$$
\big\|\langle Q\rangle^{-s}\big(R^0(\lambda+\zeta)-R^0(\lambda)\big)
\langle Q\rangle^{-s}\big\|_{{\rm HS}}
\le{\rm Const.}\;\!|\zeta|\;\!|\lambda|^{-1/2},
$$
where the constant may depend on $\varepsilon$ but not on $\lambda$ and $\zeta$.
\end{enumerate}
\end{Lemma}

\begin{proof}
The first claim follows from \eqref{eq_noyau_1}. For the second one, one has to
compute the integral kernel of
$
\langle Q\rangle^{-s}\big(R^0(\lambda+\zeta)-R^0(\lambda)\big)\langle Q\rangle^{-s}
$,
taking into account the following equalities with $y=|x-x'|$ and $x,x'\in\R:$
$$
\frac{\e^{i\sqrt{\lambda+\zeta}\;\!y}}{\sqrt{\lambda+\zeta}}
-\frac{\e^{i\sqrt\lambda\;\!y}}{\sqrt\lambda}
=\frac{-\zeta}{\sqrt\lambda\;\!\sqrt{\lambda+\zeta}
\;\!(\sqrt{\lambda+\zeta}+\sqrt\lambda)}\;\! \e^{i\sqrt{\lambda+\zeta}\;\!y}
+\frac1{\sqrt\lambda}\big(\e^{i\sqrt{\lambda+\zeta}\;\!y}-\e^{i\sqrt\lambda\;\!y}\big)
$$
and
$$
\frac1{\sqrt\lambda}\big(\e^{i\sqrt{\lambda+\zeta}\;\!y}-\e^{i\sqrt\lambda\;\!y}\big)
=\frac{i\;\!\zeta\;\!y}{2\sqrt\lambda}
\int_0^1\frac{\e^{i\sqrt{\lambda+s\;\!\zeta}\;\!y}}{\sqrt{\lambda+s\zeta}}\,\d s.
$$
\end{proof}

Now, we consider a self-adjoint operator $H:=H_0+V$, where $V\in\linf(\Omega;\R)$ is
measurable with bounded support. We impose the boundedness of the support for
simplicity, but we note that our results would also hold for potentials $V$ decaying
sufficiently fast at infinity (see for example the seminal papers \cite{JK79,JN01}
for precise conditions on the decay of $V$ at infinity). Following the standard idea
of decomposing the perturbation into factors, we define the functions
$$
v:\Omega\to\R,\quad(\omega,x)\mapsto|V(\omega,x)|^{1/2}
\qquad\hbox{and}\qquad
u:\Omega\to\{-1,1\},\quad(\omega,x)\mapsto
\begin{cases}
1  & \hbox{if}~~V(\omega,x)\ge0\\
-1 & \hbox{if}~~V(\omega,x)<0.
\end{cases}
$$
Then, the operator $u+vR_0(z)\;\!v$ has a bounded inverse in $\H$ for each
$z\in\C\setminus\R$ and the resolvent equation may be written as
$$
(H-z)^{-1}=R_0(z)-R_0(z)\;\!v\big(u+vR_0(z)\;\!v\big)^{-1}vR_0(z),
\quad z\in\C\setminus\R.
$$
Since the following equality holds:
\begin{equation}\label{resolv_eq}
uv(H-z)^{-1}vu=u-\big(u+v R_0(z)\;\!v\big)^{-1},\quad z\in\C\setminus\R,
\end{equation}
deriving expansions in $z$ for the resolvent $(H-z)^{-1}$ amounts
to deriving expansions in $z$ for the operator $\big(u+v R_0(z)v\big)^{-1}$, as we
shall do in the section.

\subsection{Asymptotic expansion at embedded thresholds or eigenvalues}

We derive in this section an asymptotic expansion in $z$ for the operator
$\big(u+v R_0(z)v\big)^{-1}$. As a by-product, we show the absence of accumulation of
eigenvalues of $H$. For this, we first adapt a convention of \cite{JN01} by
considering values $z=\lambda-\kappa^2$ with $\kappa$ belonging to the sets
$$
O(\varepsilon)
:=\big\{\kappa\in\C\mid|\kappa|\in(0,\varepsilon),~\re(\kappa)>0\hbox{ and }
\im(\kappa)<0\big\},\quad\varepsilon>0,
$$
and
$$
\widetilde O(\varepsilon)
:=\big\{\kappa\in\C\mid|\kappa|\in(0,\varepsilon),~\re(\kappa)\ge0\hbox{ and }
\im(\kappa)\le0\big\},\quad\varepsilon>0.
$$
Also, we note that if $\kappa\in O(\varepsilon)$, then $-\kappa^2\in\C_+$, while if
$\kappa\in\widetilde O(\varepsilon)$, then $-\kappa^2\in \overline{\C_+}$.

Then, the main result of this section reads as follows\;\!:

\begin{Proposition}\label{Prop_Asymp}
Suppose that $V\in\linf(\Omega;\R)$ has bounded support, let
$\lambda\in\tau\cup\sigma_{\rm p}(H)$, and take $\kappa\in O(\varepsilon)$ with
$\varepsilon>0$ small enough. Then, the operator
$\big(u+vR_0(\lambda-\kappa^2)\;\!v\big)^{-1}$ belongs to $\B(\H)$
and is continuous in $\kappa\in O(\varepsilon)$. Moreover, the continuous function
$$
O(\varepsilon)\ni\kappa\mapsto
\big(u+vR_0(\lambda-\kappa^2)\;\!v\big)^{-1}\in\B(\H)
$$
extends continuously to a function
$\widetilde O(\varepsilon)\ni\kappa\mapsto\M(\lambda,\kappa)\in\B(\H)$,
and for each $\kappa\in\widetilde O(\varepsilon)$ the
operator $\M(\lambda,\kappa)$ admits an asymptotic expansion in $\kappa$.
The precise form of this expansion is given in equations \eqref{eq_expansion_1}
and \eqref{eq_expansion_2} below.
\end{Proposition}

\begin{proof}
For each $\lambda\in\R$, $\varepsilon>0$ and $\kappa\in O(\varepsilon)$, one
has $\im(\lambda-\kappa^2)\ne0$. Thus, \eqref{resolv_eq} implies that the operator
$\big(u+vR_0(\lambda-\kappa^2)\;\!v\big)^{-1}$ belongs to $\B\big(\H)$
and is continuous in $\kappa\in O(\varepsilon)$.
For the other claims, we distinguish the cases $\lambda\in\tau$ and
$\lambda\in\sigma_{\rm p}(H)\setminus\tau$, treating first the case $\lambda\in\tau$.
All the operators defined below depend on the choice of $\lambda$, but for simplicity
we do not always mention these dependencies.

(i) Assume that $\lambda\in\tau$, take $\varepsilon>0$,
set $N:=\{n\ge1\mid\lambda_n=\lambda\}$, and write
$\P:=\sum_{n\in N}\P_n$ for the corresponding orthogonal projection (of dimension
greater or equal to $1$). Then, \eqref{factor_res} implies for
$\kappa\in O(\varepsilon)$ that
$$
\big(u+vR_0(\lambda-\kappa^2)\;\!v\big)^{-1}
=\left\{v\big(\P\otimes R^0(-\kappa^2)\big)v+u
+\sum_{n\notin N}v\big(\P_n\otimes R^0(\lambda-\kappa^2-\lambda_n)\big)v\right\}^{-1}.
$$
Moreover, the expansion
$
R^0(-\kappa^2)(x,x')
=\frac1{2\kappa}-\frac{|x-x'|}2+\kappa\;\!\frac{|x-x'|^2}4+\O(\kappa^2)
$
for $\kappa\in\widetilde O(\varepsilon)$ (see \eqref{eq_noyau_1}) implies
that the continuous function
$$
O(\varepsilon)\ni\kappa\mapsto v\big(\P\otimes R^0(-\kappa^2)\big)v\in\B(\H)
$$
extends continuously to a function
$
\widetilde O(\varepsilon)\ni\kappa\mapsto
\frac1{2\kappa}\;\!N_0+N_1(\kappa)\in\B(\H)
$
with $N_0,N_1(\kappa)\in\B(\H)$ integral operators which kernels satisfy
\begin{align*}
N_0(\omega,x,\omega',x')
&=\sum_{n\in N}f_n(\omega)\;\!v(\omega,x)\;\!v(\omega',x')\;\!
\overline{f_n(\omega')},\quad(\omega,x),(\omega',x')\in\Omega,\\
N_1(0)(\omega,x,\omega',x')
&=-\frac1{2}\sum_{n\in N}f_n(\omega)\;\!v(\omega,x)\;\!|x-x'|\;\!v(\omega',x')
\;\!\overline{f_n(\omega')},\quad(\omega,x),(\omega',x')\in\Omega.
\end{align*}
Also, Lemma \ref{lemsauve}(a) implies the existence and the unicity in $\B(\H)$ of the
limits
$$
\sum_{n\notin N}v\big(\P_n\otimes R^0(\lambda-\kappa^2-\lambda_n)\big)v
:=\lim_{\kappa'\to\kappa,\,\kappa'\in O(\varepsilon)}
\sum_{n\notin N}v\big(\P_n\otimes R^0(\lambda-\kappa'^2-\lambda_n)\big)v,
\quad\kappa\in\widetilde O(\varepsilon).
$$
Therefore, one has for $\kappa\in O(\varepsilon)$ that
$$
\big(u+vR_0(\lambda-\kappa^2)\;\!v\big)^{-1}=2\kappa\;\!I_0(\kappa)^{-1},
$$
with the operators
\begin{equation}\label{form_I_0}
I_0(\kappa):=N_0+2\kappa\;\!M_1(\kappa)
\quad\hbox{and}\quad
M_1(\kappa):=N_1(\kappa)+u+\sum_{n\notin N}
v\big(\P_n \otimes R^0(\lambda-\kappa^2-\lambda_n)\big)\;\!v
\end{equation}
continuous as functions from $\widetilde O(\varepsilon)$ to $\B(\H)$. Furthermore, one
infers from \cite[Lemma 5.1(i)]{JN01} and Lemma \ref{lemsauve}(a) that
$\|M_1(\kappa)\|_{\B(\H)}$ is uniformly bounded as $\kappa\to0$.

Our goal thus reduces to derive an asymptotic expansion for $I_0(\kappa)^{-1}$ as
$\kappa\to0$. Since $I_0(0)=N_0$ is a finite-rank operator, $0$ is not a limit point
of $\sigma(N_0)$. Also, $N_0$ is self-adjoint, therefore the orthogonal projection
$S_0$ on $\ker(N_0)$ is equal to the Riesz projection of $N_0$ associated with the
value $0$. We can thus apply Proposition \ref{propourversion}, and obtain for
$\kappa\in\widetilde O(\varepsilon)$ with $\varepsilon>0$ small enough that the operator
$I_1(\kappa):S_0\H\to S_0\H$ defined by
\begin{equation}\label{defI1}
I_1(\kappa)
:=\sum_{j\ge0}(-2\kappa)^jS_0\;\!\big\{M_1(\kappa)\big(I_0(0)+S_0\big)^{-1}\big\}^{j+1}S_0
\end{equation}
is uniformly bounded as $\kappa\to0$. Furthermore, $I_1(\kappa)$ is invertible in
$S_0\H$ with bounded inverse satisfying the equation
$$
I_0(\kappa)^{-1}
=\big(I_0(\kappa)+S_0\big)^{-1}+\frac 1{2\kappa}\;\!\big(I_0(\kappa)+S_0\big)^{-1}
S_0I_1(\kappa)^{-1}S_0\big(I_0(\kappa)+S_0\big)^{-1}.
$$
It follows that for $\kappa\in O(\varepsilon)$ with $\varepsilon>0$ small
enough, one has
\begin{equation}\label{eq18}
\big(u+vR_0(\lambda-\kappa^2)\;\!v\big)^{-1}
=2\kappa\;\!\big(I_0(\kappa)+S_0\big)^{-1}+
\big(I_0(\kappa)+S_0\big)^{-1}S_0I_1(\kappa)^{-1}S_0\big(I_0(\kappa)+S_0\big)^{-1},
\end{equation}
with the first term vanishing as $\kappa \to0$.

To describe the second term of $\big(u+vR_0(\lambda-\kappa^2)\;\!v\big)^{-1}$ as $\kappa\to0$,
we recall the equality $\big(I_0(0)+S_0\big)^{-1}S_0=S_0$, which
(together with \eqref{defI1}) implies for $\kappa\in\widetilde O(\varepsilon)$ with
$\varepsilon>0$ small enough that
\begin{equation}\label{form_I_1}
I_1(\kappa)=S_0M_1(0)S_0+\kappa\;\!M_2(\kappa),
\end{equation}
with
\begin{align}
M_2(\kappa)
&:=\frac1\kappa\;\!S_0\big(M_1(\kappa)-M_1(0)\big)S_0
+\frac1\kappa\sum_{j\ge1}(-2\kappa)^jS_0\;\!
\big\{M_1(\kappa)\big(I_0(0)+S_0\big)^{-1}\big\}^{j+1}S_0\nonumber\\
&\equiv S_0N_2(\kappa)S_0+\frac1\kappa\;\!S_0\sum_{n\notin N}
v\;\!\big\{\P_n\otimes\big(R^0(\lambda-\kappa^2-\lambda_n)
-R^0(\lambda-\lambda_n)\big)\big\}\;\!vS_0\nonumber\\
&\quad-2\sum_{j\ge0}(-2\kappa)^jS_0\;\!
\big\{M_1(\kappa)\big(I_0(0)+S_0\big)^{-1}\big\}^{j+2}S_0\label{terme2}
\end{align}
and
$$
N_2(\kappa):=\frac1\kappa\big(N_1(\kappa)-N_1(0)\big).
$$
Then, we observe that \cite[Lemma 5.1(i)]{JN01} implies that $N_2(\kappa)$ admits a
finite limit as $\kappa\to0$.
Also, we note that Lemma \ref{lemsauve}(b) implies that the second
term in \eqref{terme2} vanishes as $\kappa\to0$. Therefore, $\|M_2(\kappa)\|_{\B(S_0\H)}$
is uniformly bounded as $\kappa\to0$.

Now, we recall that
\begin{equation*}
M_1(0)=N_1(0)+u+\sum_{n\notin N}v\big(\P_n\otimes R^0(\lambda-\lambda_n)\big)\;\!v,
\end{equation*}
with $u$ unitary and self-adjoint, $N_1(0)$ self-adjoint and compact, and with the last term compact with non-negative imaginary part (the last property holds for
weighted resolvents on the real axis). So, since $S_0$ is an orthogonal projection with
finite-dimensional kernel, the operator $I_1(0)=S_0M_1(0)S_0$ acting in the Hilbert
space $S_0\H$ can also be written as the sum of a unitary and self-adjoint operator,
a self-adjoint and compact operator, and a compact operator with non-negative
imaginary part. Thus, Corollary \ref{Corol_So} applies with $S_1$ the finite-rank
orthogonal projection on $\ker\big(I_1(0)\big)$, and the iterative procedure of
Section \ref{Sec_Inv} can be applied to $I_1(\kappa)$ as it was done for
$I_0(\kappa)$.

Thus, for $\kappa\in\widetilde O(\varepsilon)$ with $\varepsilon>0$ small enough, the
operator $I_2(\kappa):S_1\H\to S_1\H$ defined by
$$
I_2(\kappa)
:=\sum_{j\ge0}(-\kappa)^jS_1\big\{M_2(\kappa)\big(I_1(0)+S_1\big)^{-1}\big\}^{j+1}S_1
$$
is uniformly bounded as $\kappa\to0$. Furthermore, $I_2(\kappa)$ is invertible in
$S_1\H$ with bounded inverse satisfying the equation
$$
I_1(\kappa)^{-1}
=\big(I_1(\kappa)+S_1\big)^{-1}+\frac1\kappa\;\!\big(I_1(\kappa)+S_1\big)^{-1}
S_1I_2(\kappa)^{-1}S_1\big(I_1(\kappa)+S_1\big)^{-1}.
$$
This expression for $I_1(\kappa)^{-1}$ can now be inserted in \eqref{eq18} in order
to get for $\kappa\in O(\varepsilon)$ with $\varepsilon>0$ small enough
\begin{align}
&\big(u+vR_0(\lambda-\kappa^2)\;\!v\big)^{-1}\nonumber\\
&=2\kappa\;\!\big(I_0(\kappa)+S_0\big)^{-1}
+\big(I_0(\kappa)+S_0\big)^{-1}S_0\big(I_1(\kappa)+S_1\big)^{-1}S_0
\big(I_0(\kappa)+S_0\big)^{-1}\nonumber\\
&\quad+\frac1\kappa\;\!\big(I_0(\kappa)+S_0\big)^{-1}
S_0\big(I_1(\kappa)+S_1\big)^{-1}S_1I_2(\kappa)^{-1}S_1\big(I_1(\kappa)+S_1\big)^{-1}
S_0\big(I_0(\kappa)+S_0\big)^{-1},\label{eq_F_second}
\end{align}
with the first two terms bounded as $\kappa\to0$.

Let us concentrate on the last term and check once more that the assumptions of
Proposition \ref{propourversion} are satisfied. For that purpose, we recall that
$\big(I_1(0)+S_1\big)^{-1}S_1=S_1$, and observe that for
$\kappa\in\widetilde O(\varepsilon)$ with $\varepsilon>0$ small enough
\begin{equation}\label{form_I_2}
I_2(\kappa)=S_1M_2(0)S_1+\kappa\;\!M_3(\kappa),
\end{equation}
with
\begin{equation}\label{devM20}
M_2(0)=S_0N_2(0)S_0 -2 S_0M_1(0)\big(I_0(0)+S_0\big)^{-1}M_1(0)S_0
\qquad\hbox{and}\qquad
M_3(\kappa)\in\O(1).
\end{equation}
The inclusion $M_3(\kappa)\in\O(1)$ follows from standard estimates and from the fact
that $\frac1\kappa\big(N_2(\kappa)-N_2(0)\big)$ admits a finite limit as $\kappa\to0$
(see \cite[Lemma 5.1(i)]{JN01}).
Note also that the kernel of $N_2(0)$ is given by
\begin{equation}\label{N20}
N_2(0)(\omega,x,\omega',x')
=\frac14\sum_{n\in N}f_n(\omega)\;\!v(\omega,x)\;\!|x-x'|^2\;\!v(\omega',x')
\;\!\overline{f_n(\omega')},\quad(\omega,x),(\omega',x')\in\Omega.
\end{equation}
Now, as already observed, one has $M_1(0)=X+iZ^*Z$, with $X,Z$ bounded self-adjoint
operators in $\H$. Therefore it follows that
$I_1(0)=S_0M_1(0)S_0=S_0XS_0+i(ZS_0)^*(Z S_0)$, and one infers from Corollary
\ref{Cor_magique} that $Z S_0S_1=0$ and $S_1S_0 Z^*=0$. Since $S_1S_0=S_1=S_0S_1$, it
follows that $ZS_1=0$, that $S_1Z^*=0$, and also that
\begin{align*}
S_1M_1(0)\big(I_0(0)+S_0\big)^{-1}M_1(0)S_1
&=S_1(X+iZ^*Z)\big(I_0(0)+S_0\big)^{-1}(X+iZ^*Z)S_1\\
&=S_1X\big(I_0(0)+S_0\big)^{-1}XS_1.
\end{align*}
So, this operator is self-adjoint, and thus one infers from \eqref{devM20} and
\eqref{N20} that $I_2(0)=S_1M_2(0)S_1$ is the sum of two bounded self-adjoint
operators in $S_1\H$.

Since $S_1\H$ is finite-dimensional, $0$ is not a limit point of the spectrum of
$I_2(0)$. So, the orthogonal projection $S_2$ on $\ker\big(I_2(0)\big)$ is a
finite-rank operator, and Proposition \ref{propourversion} applies to
$I_2(0)+\kappa\;\!M_3(\kappa)$. Thus, for $\kappa\in\widetilde O(\varepsilon)$ with
$\varepsilon>0$ small enough, the operator $I_3(\kappa):S_2\H\to S_2\H$ defined by
$$
I_3(\kappa)
:=\sum_{j\ge0}(-\kappa)^jS_2\;\!\big\{M_3(\kappa)\big(I_2(0)+S_2\big)^{-1}\big\}^{j+1}S_2
$$
is uniformly bounded as $\kappa\to0$. Furthermore, $I_3(\kappa)$ is invertible in
$S_2\H$ with bounded inverse satisfying the equation
$$
I_2(\kappa)^{-1}
=\big(I_2(\kappa)+S_2\big)^{-1}
+\frac1\kappa\;\!\big(I_2(\kappa)+S_2\big)^{-1}S_2I_3(\kappa)^{-1}S_2
\big(I_2(\kappa)+S_2\big)^{-1}.
$$
This expression for $I_2(\kappa)^{-1}$ can now be inserted in \eqref{eq_F_second} in
order to get for $\kappa\in O(\varepsilon)$ with $\varepsilon>0$ small enough
\begin{align}
&\big(u+vR_0(\lambda-\kappa^2)\;\!v\big)^{-1}\nonumber\\
&=2\kappa\big(I_0(\kappa)+S_0\big)^{-1}
+\big(I_0(\kappa)+S_0\big)^{-1}S_0\big(I_1(\kappa)+S_1\big)^{-1}S_0
\big(I_0(\kappa)+S_0\big)^{-1}\nonumber\\
&\quad+\frac1\kappa\;\!\big(I_0(\kappa)+S_0\big)^{-1}S_0
\big(I_1(\kappa)+S_1\big)^{-1}S_1\big(I_2(\kappa)+S_2\big)^{-1}S_1
\big(I_1(\kappa)+S_1\big)^{-1}S_0\big(I_0(\kappa)+S_0\big)^{-1}\nonumber\\
&\quad+\frac1{\kappa^2}\;\!\big(I_0(\kappa)+S_0\big)^{-1}S_0
\big(I_1(\kappa)+S_1\big)^{-1}S_1\big(I_2(\kappa)+ S_2\big)^{-1}S_2 I_3(\kappa)^{-1}
S_2\big(I_2(\kappa)+S_2\big)^{-1}S_1\nonumber\\
&\qquad\times\big(I_1(\kappa)+S_1\big)^{-1}S_0\big(I_0(\kappa)+S_0\big)^{-1}.
\label{sol1}
\end{align}

Fortunately, the iterative procedure stops here. The argument is based on the relation
$$
uv\;\!(H-\lambda+\kappa^2)^{-1}vu= u-\big(u+vR_0(\lambda-\kappa^2)\;\!v\big)^{-1}
$$
and the fact that $H$ is a self-adjoint operator. Indeed, if we choose
$\kappa=\frac\varepsilon2(1-i)\in O(\varepsilon)$, then the inequality
$\big\|\kappa^2(H-\lambda+\kappa^2)^{-1}\big\|_{\B(\H)}\le1$ holds, and thus
\begin{equation}\label{notresauveur}
\limsup_{\kappa\to0}
\big\|\kappa^2\big(u+vR_0(\lambda-\kappa^2)\;\!v\big)^{-1}\big\|_{\B(\H)}<\infty.
\end{equation}
So, if we replace $\big(u+vR_0(\lambda-\kappa^2)\;\!v\big)^{-1}$ by the expression
\eqref{sol1} and if we take into account that all factors of the form
$\big(I_j(\kappa)+S_j\big)^{-1}$ have a finite limit as $\kappa\to0$, we infer from
\eqref{notresauveur} that
\begin{equation}\label{notresecondsauveur}
\limsup_{\kappa\to0}\big\|I_3(\kappa)^{-1}\big\|_{\B(S_2\H)}<\infty.
\end{equation}
Therefore, it only remains to show that this relation holds not just for
$\kappa=\frac\varepsilon2(1-i)$ but for all $\kappa\in\widetilde O(\varepsilon)$. For
that purpose, we consider $I_3(\kappa)$ once again, and note that
\begin{equation}\label{souffrance1}
I_3(\kappa)=S_2M_3(0)S_2+\kappa\;\!M_4(\kappa)
\quad\hbox{with}\quad M_4(\kappa)\in\O(1).
\end{equation}
The precise form of $M_3(0)$ can be computed explicitly, but is irrelevant. Now, since
$I_3(0)$ acts in a finite-dimensional space, $0$ is an isolated eigenvalue of $I_3(0)$
if $0\in\sigma\big(I_3(0)\big)$, in which case we write $S_3$ for the corresponding
Riesz projection. Then, the operator $I_3(0)+S_3$ is invertible with bounded inverse,
and \eqref{souffrance1} implies that $I_3(\kappa)+S_3$ is also invertible with bounded
inverse for $\kappa\in\widetilde O(\varepsilon)$ with $\varepsilon>0$ small enough. In
addition, one has
$\big(I_3(\kappa)+S_3\big)^{-1}=\big(I_3(0)+S_3\big)^{-1}+\O(\kappa)$. By the
inversion formula given in \cite[Lemma 2.1]{JN01}, one infers that
$S_3-S_3\big(I_3(\kappa)+S_3\big)^{-1}S_3$ is invertible in $S_3\H$ with bounded
inverse and that the following equalities hold
\begin{align*}
I_3(\kappa)^{-1}
&=\big(I_3(\kappa)+S_3\big)^{-1}+\big(I_3(\kappa)+S_3\big)^{-1}S_3
\big\{S_3-S_3\big(I_3(\kappa)+S_3\big)^{-1}S_3\big\}^{-1}S_3
\big(I_3(\kappa)+S_3\big)^{-1}\\
&=\big(I_3(\kappa)+S_3\big)^{-1}+\big(I_3(\kappa)+S_3\big)^{-1}S_3
\big\{S_3-S_3\big(I_3(0)+S_3\big)^{-1}S_3+\O(\kappa)\big\}^{-1}S_3
\big(I_3(\kappa)+S_3\big)^{-1}.
\end{align*}
This implies that \eqref{notresecondsauveur} holds for some
$\kappa\in\widetilde O(\varepsilon)$ if and only if the operator
$S_3-S_3\big(I_3(0)+S_3\big)^{-1}S_3$ is invertible in $S_3\H$ with bounded inverse.
But, we already know from what precedes that \eqref{notresecondsauveur} holds for
$\kappa=\frac\varepsilon2(1-i)$. So, the operator
$S_3-S_3\big(I_3(0)+S_3\big)^{-1}S_3$ is invertible in $S_3\H$ with bounded inverse,
and thus \eqref{notresecondsauveur} holds for all $\kappa\in\widetilde O(\varepsilon)$.

Therefore, \eqref{sol1} implies that the function
$$
O(\varepsilon)\ni\kappa\mapsto
\big(u+vR_0(\lambda-\kappa^2)\;\!v\big)^{-1}\in\B(\H)
$$
extends continuously to a function
$\widetilde O(\varepsilon)\ni\kappa\mapsto\M(\lambda,\kappa)\in\B(\H)$,
with $\M(\lambda,\kappa)$ given by
\begin{align}
\M(\lambda,\kappa)
&=2\kappa\big(I_0(\kappa)+S_0\big)^{-1}
+\big(I_0(\kappa)+S_0\big)^{-1}S_0\big(I_1(\kappa)+S_1\big)^{-1}S_0
\big(I_0(\kappa)+S_0\big)^{-1}\nonumber\\
&\quad+\frac1\kappa\;\!\big(I_0(\kappa)+S_0\big)^{-1}S_0
\big(I_1(\kappa)+S_1\big)^{-1}S_1\big(I_2(\kappa)+S_2\big)^{-1}S_1
\big(I_1(\kappa)+S_1\big)^{-1}S_0\big(I_0(\kappa)+S_0\big)^{-1}\nonumber\\
&\quad+\frac1{\kappa^2}\;\!\big(I_0(\kappa)+S_0\big)^{-1}S_0
\big(I_1(\kappa)+S_1\big)^{-1}S_1\big(I_2(\kappa)+ S_2\big)^{-1}S_2 I_3(\kappa)^{-1}
S_2\big(I_2(\kappa)+S_2\big)^{-1}S_1\nonumber\\
&\qquad\times\big(I_1(\kappa)+S_1\big)^{-1}S_0\big(I_0(\kappa)+S_0\big)^{-1}.
\label{eq_expansion_1}
\end{align}

(ii) Assume now that $\lambda\in\sigma_{\rm p}(H)\setminus\tau$, take
$\varepsilon>0$, let $\kappa\in\widetilde O(\varepsilon)$, and set
$J_0(\kappa):=T_0+\kappa^2\;\!T_1(\kappa)$ with
$$
T_0:=u+\sum_nv\big(\P_n\otimes R^0(\lambda-\lambda_n)\big)\;\!v
$$
and
$$
T_1(\kappa)
:=\frac1{\kappa^2}\sum_nv\;\!\big\{\P_n\otimes\big(R^0(\lambda-\kappa^2-\lambda_n)
-R^0(\lambda-\lambda_n)\big)\big\}\;\!v.
$$
Then, one infers from Lemma \ref{lemsauve}(b) that $\|T_1(\kappa)\|_{\B(\H)}$ is uniformly
bounded as $\kappa\to0$. Also, the assumptions of Corollary \ref{Corol_So} hold for
the operator $T_0$, the Riesz projection $S$ associated with the value
$0\in\sigma(T_0)$ is an orthogonal projection, and Proposition \ref{propourversion}
applies for $J_0(\kappa)$. It follows that for $\kappa\in\widetilde O(\varepsilon)$
with $\varepsilon>0$ small enough, the operator $J_1(\kappa):S\H\to S\H$ defined by
$$
J_1(\kappa)
:=\sum_{j\ge0}(-\kappa^2)^jS\;\!\big\{T_1(\kappa)(T_0+S)^{-1}\big\}^{j+1}S
$$
is uniformly bounded as $\kappa\to0$. Furthermore, $J_1(\kappa)$ is invertible in
$S\H$ with bounded inverse satisfying the equation
$$
J_0(\kappa)^{-1}
=\big(J_0(\kappa)+S\big)^{-1}
+\frac1{\kappa^2}\;\!\big(J_0(\kappa)+S)^{-1}SJ_1(\kappa)^{-1}S
\big(J_0(\kappa)+S\big)^{-1}.
$$
It follows that for $\kappa\in O(\varepsilon)$ with $\varepsilon>0$ small
enough one has
\begin{equation}\label{sol2}
\big(u+vR_0(\lambda-\kappa^2)\;\!v\big)^{-1}
=\big(J_0(\kappa)+S\big)^{-1}
+\frac1{\kappa^2}\;\!\big(J_0(\kappa)+S)^{-1}SJ_1(\kappa)^{-1}S
\big(J_0(\kappa)+S\big)^{-1}.
\end{equation}
Fortunately, the iterative procedure already stops here. Indeed, the argument is
similar to the one presented above once we observe that
$$
J_1(\kappa)=ST_1(0)S+\kappa\;\!T_2(\kappa)\quad\hbox{with}\quad T_2(\kappa)\in\O(1).
$$
Therefore, \eqref{sol2} implies that the function
$$
O(\varepsilon)\ni\kappa\mapsto
\big(u+vR_0(\lambda-\kappa^2)\;\!v\big)^{-1}\in\B(\H)
$$
extends continuously to a function
$\widetilde O(\varepsilon)\ni\kappa\mapsto\M(\lambda,\kappa)\in\B(\H)$, with
$\M(\lambda,\kappa)$ given by
\begin{equation}\label{eq_expansion_2}
\M(\lambda,\kappa)
=\big(J_0(\kappa)+S\big)^{-1}
+\frac1{\kappa^2}\big(J_0(\kappa)+S)^{-1}SJ_1(\kappa)^{-1}S
\big(J_0(\kappa)+S\big)^{-1}.
\end{equation}
\end{proof}

We now give a result on the possible embedded eigenvalues. Since it is already known that the eigenvalues of $H$ in $\sigma(H)\setminus\tau$ are of finite multiplicity and can accumulate
at points of $\tau$ only (see \cite[Thm.~3.4(b)]{Tie06}), we show that such accumulations
do not take place\;\!:

\begin{Corollary}\label{noaccp}
Suppose that $V\in\linf(\Omega;\R)$ has bounded support. Then, the point spectrum of
$H$ has no accumulation point (except possibly at $+\infty$).
\end{Corollary}

\begin{proof}
To show the absence of local accumulation of eigenvalues, suppose by absurd that there
is an accumulation of eigenvalues at some point $\lambda\in\tau$. Then, the validity
of the expansion \eqref{eq_expansion_1} at the point $\lambda$ contradicts the
validity of the expansion \eqref{eq_expansion_2} which would take place at each of
these eigenvalues. Thus, there is no accumulation of eigenvalues at points of $\tau$,
and the claim is proved.
\end{proof}

We end up this section with some auxiliary results which will be useful later on. All
notations and definitions are borrowed from the proof of Proposition \ref{Prop_Asymp}.
The only change is that we extend by $0$ the operators defined originally on subspaces
of $\H$ to get operators defined on all of $\H$.

\begin{Lemma}\label{com_simples}
Take $2\ge j\ge k\ge0$ and $\kappa\in\widetilde O(\varepsilon)$ with
$\varepsilon>0$ small enough. Then, one has in $\B(\H)$
$$
\big[S_j,\big(I_k(\kappa)+S_k\big)^{-1}\big]\in\O(\kappa).
$$
\end{Lemma}

\begin{proof}
The fact that $S_j$ is the orthogonal projection on the kernel of $I_j(0)$ and the
relations $S_kS_j=S_j=S_jS_k$ imply that $[S_k,S_j]=0$ and $[I_k(0),S_j]=0$. Thus,
one has the equalities
\begin{align*}
\big[S_j,\big(I_k(\kappa)+S_k\big)^{-1}\big]
&=\big(I_k(\kappa)+S_k\big)^{-1}\big[I_k(\kappa)+S_k,S_j\big]
\big(I_k(\kappa)+S_k\big)^{-1}\\
&=\big(I_k(\kappa)+S_k\big)^{-1}\big[I_k(0)+\O(\kappa)+S_k,S_j\big]
\big(I_k(\kappa)+S_k\big)^{-1}\\
&=\big(I_k(\kappa)+S_k\big)^{-1}\big[\O(\kappa),S_j\big]
\big(I_k(\kappa)+S_k\big)^{-1},
\end{align*}
which implies the claim.
\end{proof}

Given $\lambda\in\tau$, we recall that $N=\big\{n\ge1\mid\lambda_n=\lambda\big\}$ and
$\P=\sum_{n\in N}\P_n$.

\begin{Lemma}\label{relations_simples}
Let $\lambda\in\tau$ and let $\G$ be an auxiliary Hilbert space.
\begin{enumerate}
\item[(a)] For each $n\in N$, one has $(\P_n\otimes1)\;\!vS_0=0$.
\item[(b)] For each $n\notin N$ and $B_n\in\B(\H,\G)$ such that
$B_n^*B_n=\im\big\{v\big(\P_n\otimes R^0(\lambda-\lambda_n)\big)v\big\}$, one has
$S_1B_n^*=0$ and $B_nS_1=0$.
\end{enumerate}
\end{Lemma}

\begin{proof}
The first claim follows from the fact that $S_0$ is the orthogonal projection on
$\ker\big(v\;\!(\P\otimes1)\;\!v\big)$. The second claim follows from Lemma
\ref{Cor_magique} applied with $Z_n=B_nS_0$ and
$$
A_0
=S_0M_1(0)S_0
=S_0\left\{N_1(0)+u
+\sum_{n\notin N}v\big(\P_n\otimes R^0(\lambda-\lambda_n)\big)v\right\}S_0
$$
if one takes into account the relations $S_0S_1=S_1=S_1S_0$.
\end{proof}

For what follows, we recall that $Q$ is the multiplication operator by the variable in
$\ltwo(\R)$.

\begin{Lemma}\label{la_cle_des_champs}
One has
\begin{enumerate}
\item[(a)] $XS_2=0=S_2X$, with $X$ the real part of the operator $M_1(0)$,
\item[(b)] $S_2\;\!(1\otimes Q)\;\!v\;\!(f_n\otimes1)=0\,$ for all $n\in N$,
\item[(c)] $M_1(0)S_2=0=S_2M_1(0)$.
\end{enumerate}\end{Lemma}

\begin{proof}
First, we recall from the proof of Proposition \ref{Prop_Asymp} that
$$
I_2(0)
=S_1M_2(0)S_1
=S_1N_2(0)S_1-2S_1X\big(I_0(0)+S_0\big)^{-1}X S_1,
$$
with $N_2(0)$ given (in the usual bra-ket notation) by
\begin{align*}
N_2(0)
=\frac14\sum_{n\in N}\big\{
&\big|(1\otimes Q^2)\;\!v\;\!(f_n\otimes1)\big\rangle\big\langle v\;\!(f_n\otimes1)\big|
+\big|v\;\!(f_n\otimes1)\big\rangle\big\langle(1\otimes Q^2)\;\!v\;\!(f_n\otimes1)\big|\\
&-2\;\!\big|(1\otimes Q)\;\!v\;\!(f_n\otimes1)\big\rangle
\big\langle(1\otimes Q)\;\!v\;\!(f_n\otimes1)\big|\big\}.
\end{align*}
Now, let $\varphi\in S_2\H$. Then, we have $I_2(0)\varphi=0$ and
\begin{align}\label{eq_rhume}
\big\langle\varphi,N_2(0)\varphi\big\rangle
=2\;\!\big\langle\varphi,X\big(I_0(0)+S_0\big)^{-1}X\varphi\big\rangle.
\end{align}
In addition, one infers from the relation $S_2=S_0S_2$ and Lemma
\ref{relations_simples}(a) that
$$
\big\langle\varphi,\big\{\big|(1\otimes Q^2)\;\!v\;\!(f_n\otimes1)\big\rangle
\big\langle v\;\!(f_n\otimes1)\big|\big\}\varphi\big\rangle
=\big\langle\varphi,(1\otimes Q^2)v\;\!(f_n\otimes1)\big\rangle
\big\langle S_0\;\!v\;\!(f_n\otimes1),\varphi\big\rangle
=0,
$$
and thus \eqref{eq_rhume} reduces to
$$
-\left\langle\varphi,\sum_{n\in N}\big\{\big|(1\otimes Q)\;\!v\;\!(f_n\otimes1)\big\rangle
\big\langle(1\otimes Q)\;\!v\;\!(f_n\otimes1)\big|\big\}\varphi\right\rangle
=4\left\langle\varphi,X\big(I_0(0)+S_0\big)^{-1}X\varphi\right\rangle.
$$
Since both operators are positive, both sides of the equality are equal to $0$.
This implies that
$$
\big\langle(1\otimes Q)\;\!v\;\!(f_n\otimes1),\varphi\big\rangle=0
~\hbox{for each}~n\in N\quad\hbox{and}\quad
\big\|\big(I_0(0)+S_0\big)^{-1/2}X\varphi\big\|^2=0,
$$
from which the points (a) and (b) are easily deduced.

Finally, we note that $M_1(0)S_2=XS_2$ and $S_2M_1(0)=S_2X$ due to the proof of
Proposition \ref{Prop_Asymp}. So, the point (c) follows from the point (a).
\end{proof}

\subsection{Scattering theory and spectral representation}

In this section, we recall some basics on the scattering theory for the pair
$\{H_0,H\}$ and on the spectral decomposition for $H_0$. As before, we assume that
$V\in\linf(\Omega;\R)$ has bounded support.

Under this assumption, it is a well-known that the wave operators
$$
W_\pm:=\slim_{t\to\pm\infty}\e^{itH}\e^{-itH_0}
$$
exist and are complete (see \cite[Cor.~3.5(b)]{Tie06}). As a consequence, the
scattering operator $S:=W_+^*W_-$ is a unitary operator in $\H$ which commutes with
$H_0$, and thus $S$ is decomposable in the spectral representation of $H_0$. So, in
order to proceed, we start by recalling the spectral representation of $H_0$. For that
purpose, we define for each $\lambda\in[\lambda_1,\infty)$ the finite set
$$
\N(\lambda):=\big\{n\ge1 \mid \lambda_n\le\lambda\big\}
$$
and the finite-dimensional space
$$
\Hrond(\lambda)
:=\bigoplus_{n\in\N(\lambda)}
\big\{\P_n\;\!\ltwo(\Sigma)\oplus\P_n\;\!\ltwo(\Sigma)\big\},
$$
with $\lambda_n$ and $\P_n$ as in Section \ref{secwave}. Note that $\Hrond(\lambda)$
is naturally embedded in
$
\Hrond(\infty)
:=\bigoplus_{n\ge1}\big\{\P_n\;\!\ltwo(\Sigma)\oplus\P_n\;\!\ltwo(\Sigma)\big\}
$.
Now, for any $\xi\in \R$, we let $\gamma(\xi):{\mathscr S}(\R)\to\C$ be the trace
operator given by $\gamma(\xi)f=f(\xi)$, with ${\mathscr S}(\R)$ the Schwartz space on
$\R$. Also, we define for each $\lambda\in[\lambda_1,\infty)\setminus\tau$ the
operator $T(\lambda):\ltwo(\Sigma)\odot{\mathscr S}(\R)\to\Hrond(\lambda)$ by
$$
\big(T(\lambda)\;\!\varphi\big)_n
:=(\lambda-\lambda_n)^{-1/4}
\big\{\big(\P_n\otimes\gamma(-\sqrt{\lambda-\lambda_n})\big)\varphi,
\big(\P_n\otimes\gamma(\sqrt{\lambda-\lambda_n})\big)\varphi\big\},
\quad n\in\N(\lambda).
$$
Some regularity properties of the map $\lambda\mapsto T(\lambda)$ have been
established in \cite[Lemma~2.4]{Tie06}, and additional properties are derived below
for the related map $\lambda\mapsto\F_0(\lambda)$ which we now define.

Let $\F:\ltwo(\R)\to\ltwo(\R)$ be the Fourier transform and let
$\Hrond:=\int_{[\lambda_1,\infty)}^\oplus\Hrond(\lambda)\,\d\lambda$. Then, the
operator $\F_0:\H\to\Hrond$ given by
$$
(\F_0\;\!\varphi)(\lambda)
\equiv\F_0(\lambda)\;\!\varphi
:=2^{-1/2}\;\!T(\lambda)(1\otimes \F)\;\!\varphi,
\quad\lambda\in[\lambda_1,\infty)\setminus\tau,
~\varphi\in\ltwo(\Sigma)\odot{\mathscr S}(\R),
$$
is unitary and satisfies
$\F_0H_0\F_0^*=\int_{[\lambda_1,\infty)}^\oplus\lambda\,\d\lambda$ (see
\cite[Prop.~2.5]{Tie06}). We shall need some expansions for the map
$\lambda\mapsto\F_0(\lambda)$ in neighbourhoods of points
$\lambda\in\tau\cup\sigma_{\rm p}(H)$. For this, we define for each
$\lambda>\lambda_1$, each $n\ge1$ such that $\lambda_n<\lambda$, and each
$\sigma\in\{+,-\}$
$$
\F_0(\lambda;n,\sigma)\;\!\varphi
:= 2^{-1/2}(\lambda-\lambda_n)^{-1/4}
\big(\P_n\otimes \gamma(\sigma\sqrt{\lambda-\lambda_n})\F\big)\varphi,
\quad\varphi\in\ltwo(\Sigma)\odot{\mathscr S}(\R).
$$
The operator
$\F_0(\lambda;n,\sigma):\ltwo(\Sigma)\odot{\mathscr S}(\R)\to\P_n\;\!\ltwo(\Sigma)$ is
defined on a slightly larger set of $\lambda$ than the operator
$\F_0(\lambda):\ltwo(\Sigma)\odot{\mathscr S}(\R)\to\Hrond(\lambda)$. Also, we define
the sets
$$
\partial O(\varepsilon)
:=\big\{\kappa\in\C\mid\kappa\in(0,\varepsilon)\cup(0,-i\varepsilon)\big\}
\subset\widetilde O(\varepsilon),
\quad\varepsilon>0,
$$
for which $-\kappa^2\in(-\varepsilon^2,\varepsilon^2)\setminus\{0\}$ if
$\kappa\in\partial O(\varepsilon)$, and we let $\ltwo_s(\R)$ be the domain of
$\langle Q\rangle^s$, $s\in\R$, endowed with the graph norm. Then, given
$\lambda\in\tau\cup\sigma_{\rm p}(H)$, we consider for each
$\kappa\in\partial O(\varepsilon)$ with $\varepsilon>0$ small enough the asymptotic
expansion in $\kappa$ of the operator $\F_0(\lambda-\kappa^2;n,\sigma)$. If
$\lambda_n<\lambda$, one has for $\kappa\in\partial O(\varepsilon)$ with
$\varepsilon>0$ small enough
$$
(\lambda-\kappa^2-\lambda_n)^{-1/4}
=(\lambda-\lambda_n)^{-1/4}\left(1+\frac{\kappa^2}{4(\lambda-\lambda_n)}
+\O(\kappa^4)\right).
$$
Similarly, if $s>0$ is big enough and if $\sigma\in\{+,-\}$, one has in
$\B\big(\ltwo_s(\R),\C\big)$
$$
\gamma(\sigma\sqrt{\lambda-\kappa^2-\lambda_n})\;\!\F
=\gamma(\sigma\sqrt{\lambda-\lambda_n})\;\!\F
\left(1+\frac{i\sigma\kappa^2}{2\sqrt{\lambda-\lambda_n}}\;\!Q\right)+\O(\kappa^4).
$$
As a consequence, we have in
$\B\big(\ltwo(\Sigma)\otimes\ltwo_s(\R);\P_n\;\!\ltwo(\Sigma)\big)$
\begin{equation}\label{dev1}
\F_0(\lambda-\kappa^2;n,\sigma)
=\F_0(\lambda;n,\sigma)\left(1+\frac{\kappa^2}{4(\lambda-\lambda_n)}
+\frac{i\sigma\kappa^2}{2\sqrt{\lambda-\lambda_n}}\;\!Q\right)+\O(\kappa^4).
\end{equation}
On the other hand, if $\lambda=\lambda_n\in\tau$ and $-\kappa^2>0$,
then one obtains in $\B\big(\ltwo(\Sigma)\otimes\ltwo_s(\R),\P_n\;\!\ltwo(\Sigma)\big)$
\begin{equation}\label{dev2}
\F_0(\lambda-\kappa^2;n,\sigma)
=(-\kappa^2)^{-1/4}\;\!\gamma_0(n)-i\sigma(-\kappa^2)^{1/4}\;\!\gamma_1(n)
+\O(|\kappa|^{3/2})
\end{equation}
with $\gamma_j(n):\ltwo(\Sigma)\otimes\ltwo_s(\R)\to\P_n\;\!\ltwo(\Sigma)$ the
operator given by
$$
\big(\gamma_j(n)\varphi\big)(\omega)
:=\frac1{2\;\!j!\sqrt\pi}\int_\R x^j\big((\P_n\otimes1)\varphi\big)(\omega,x)\,\d x
\quad\hbox{for almost every }\omega\in\Sigma.
$$

With these expansions at hand, we can start the study of the regularity properties of
the scattering matrix at thresholds or at embedded eigenvalues. Before that, we just
need to give a final auxiliary result. Recall that the orthogonal projections $S_0$
and $S_1$ have been introduced in the proof of Proposition \ref{Prop_Asymp}.

\begin{Lemma}\label{help1}
Take $\lambda\in\tau$, $\sigma\in\{+,-\}$, and $\kappa\in\partial O(\varepsilon)$ with
$\varepsilon>0$ small enough.
\begin{enumerate}
\item[(a)] For $n\ge1$ such that $\lambda_n<\lambda$, one has
$\F_0(\lambda-\kappa^2;n,\sigma)\;\!vS_1\in\O(\kappa^2)$.
\item[(b)] For $n\ge1$ such that $\lambda_n=\lambda$ and for $-\kappa^2>0$, one has
$\F_0(\lambda-\kappa^2;n,\sigma)\;\!vS_0=0$.
\end{enumerate}
\end{Lemma}

\begin{proof}
(a) Due to the expansion \eqref{dev1}, it is sufficient to show the equality
$\F_0(\lambda;n,\sigma)vS_1=0$. For that
purpose, we define the operator $B_n:\H\to\P_n\;\!\ltwo(\Sigma)\oplus\P_n\;\!\ltwo(\Sigma)$
by
$$
B_n\;\!\varphi
:=\pi^{1/2}\big\{\F_0(\lambda;n,-)\;\!v\;\!\varphi,
\F_0(\lambda;n,+)\;\!v\;\!\varphi\big\},
$$
and note that $B_n^*B_n=\im\big\{v\big(\P_n\otimes R^0(\lambda-\lambda_n)\big)v\big\}$.
The mentioned equality then follows from Lemma \ref{relations_simples}(b).

(b) The claim is a direct consequence of the identity
$$
\F_0(\lambda-\kappa^2;n,\sigma)\;\!vS_0
=\F_0(\lambda-\kappa^2;n,\sigma)(\P_n\otimes1)\;\!vS_0
$$
and Lemma \ref{relations_simples}(a).
\end{proof}

\subsection{Continuity of the scattering matrix}

Since the scattering operator $S$ commutes with $H_0$, it follows from the spectral
decomposition of $H_0$ that
$$
\F_0\;\!S\;\!\F_0^*=\int_{[\lambda_1,\infty)}^\oplus S(\lambda)\,\d\lambda,
$$
where $S(\lambda)$, the scattering matrix at energy $\lambda$, is defined and is a
unitary operator in $\Hrond(\lambda)$ for almost every
$\lambda\in[\lambda_1,\infty)$. In addition, one can obtain a convenient stationary
formula for $S(\lambda)$ using time-dependent scattering theory. For instance, if one
uses the results of \cite[Sec.~3.1]{Tie06} and relation \eqref{resolv_eq}, one
obtains for each $\lambda\in[\lambda_1,\infty)\setminus\{\tau\cup\sigma_{\rm p}(H)\}$
the equality in $\B\big(\Hrond(\lambda)\big)$
$$
S(\lambda)
=1-2\pi i\;\!\F_0(\lambda)\;\!v\big(u+vR_0(\lambda)\;\!v\big)^{-1}v\;\!\F_0(\lambda)^*,
$$
and that the map
$$
[\lambda_1,\infty)\setminus\{\tau\cup\sigma_{\rm p}(H)\}\ni\lambda\mapsto S(\lambda)
\in\Hrond(\infty)
$$
is a $k$-times continuously differentiable, for any $k\ge0$.

Since the regularity of the map $\lambda\mapsto S(\lambda)$ is already known when
$\lambda\in[\lambda_1,\infty)\setminus\{\tau \cup \sigma_{\rm p}(H)\}$, we now
describe the behavior of $S(\lambda)$ as $\lambda$ approaches points of
$\tau\cup\sigma_{\rm p}(H)$. To do this, we decompose the scattering matrix
$S(\lambda)$ into a collection of channel scattering matrices corresponding to the
transverse modes of the waveguide. Namely, for
$\lambda\in[\lambda_1,\infty)\setminus\{\tau\cup\sigma_{\rm p}(H)\}$, for $n,n'\ge1$
such that $\lambda_n<\lambda$ and $\lambda_{n'}<\lambda$, and for
$\sigma,\sigma'\in\{+,-\}$, we define the operators
$
S(\lambda;n,\sigma,n',\sigma')
\in\B\big(\P_{n'}\;\!\ltwo(\Sigma),\P_n\;\!\ltwo(\Sigma)\big)
$
by
$$
S(\lambda;n,\sigma,n',\sigma')
:=\delta_{n\sigma n'\sigma'}-2\pi i\;\!\F_0(\lambda;n,\sigma)\;\!v
\big(u+vR_0(\lambda)\;\!v\big)^{-1}v\;\!\F_0(\lambda;n',\sigma')^*
$$
with $\delta_{n\sigma n'\sigma'}:=1$ if $(n,\sigma)=(n',\sigma')$, and
$\delta_{n\sigma n'\sigma'}:=0$ otherwise.

We consider separately the continuity at thresholds and the continuity at embedded
eigenvalues, starting with the thresholds. Note that for each $\lambda\in\tau$, a
channel can either be already open (in which case one has to show the existence and
the equality of the limits from the right and from the left), or can open at the
energy $\lambda$ (in which case one has only to show the existence of the limit from
the right).

\begin{Proposition}\label{propcon}
Suppose that $V\in\linf(\Omega;\R)$ has bounded support and take $\lambda\in\tau$,
$\kappa\in\partial O(\varepsilon)$ with $\varepsilon>0$ small enough, $n,n'\ge1$, and
$\sigma,\sigma'\in\{+,-\}$.
\begin{enumerate}
\item[(a)] If $\lambda_n<\lambda$ and  $\lambda_{n'}<\lambda$, then the limit
$\lim_{\kappa\to0}S(\lambda-\kappa^2;n,\sigma,n',\sigma')$ exists.
\item[(b)] If $\lambda_n\le\lambda$, $\lambda_{n'}\le\lambda$ and $-\kappa^2>0$, then
the limit $\lim_{\kappa\to0}S(\lambda-\kappa^2 ;n,\sigma,n',\sigma')$ exists.
 \end{enumerate}
\end{Proposition}

Before giving the proof, we define for $2\ge j\ge k\ge0$ the operators
$$
C_{jk}(\kappa):=\big[S_j,\big(I_k(\kappa)+S_k\big)^{-1}\big]\in\B(\H).
$$
We know from Lemma \ref{com_simples} that $C_{jk}(\kappa)\in\O(\kappa)$, but the
formulas \eqref{form_I_0}, \eqref{form_I_1} and \eqref{form_I_2} imply in fact that
$C_{jk}'(0):=\lim_{\kappa\to0}\frac1\kappa\;\!C_{jk}(\kappa)$ exists in $\B(\H)$. In
other cases, we use the notation $F(\kappa)\in\Oa(\kappa^n)$ for an operator
$F(\kappa)\in\O(\kappa^n)$ such that $\lim_{\kappa \to 0}\kappa^{-n}F(\kappa)$ exists
in $\B(\H)$. Finally, we note that \eqref{eq_expansion_1} can be rewritten as
\begin{align}
&\M(\lambda,\kappa)\nonumber\\
&=2\kappa\big(I_0(\kappa)+S_0\big)^{-1}\nonumber\\
&\quad+\Big(S_0\big(I_0(\kappa)+S_0\big)^{-1}-C_{00}(\kappa)\Big)
S_0\big(I_1(\kappa)+S_1\big)^{-1}S_0
\Big(\big(I_0(\kappa)+S_0\big)^{-1}S_0+C_{00}(\kappa)\Big)\nonumber\\
&\quad+\frac1\kappa\big(I_0(\kappa)+S_0\big)^{-1}
\Big(S_1\big(I_1(\kappa)+S_1\big)^{-1}-S_0C_{11}(\kappa)\Big)
S_1\big(I_2(\kappa)+S_2\big)^{-1}S_1\nonumber\\
&\qquad\times\Big(\big(I_1(\kappa)+S_1\big)^{-1}S_1+C_{11}(\kappa)S_0\Big)
\big(I_0(\kappa)+S_0\big)^{-1}\nonumber\\
&\quad+\frac1{\kappa^2}\big(I_0(\kappa)+S_0\big)^{-1}S_0
\big(I_1(\kappa)+S_1\big)^{-1}
\Big(S_2\big(I_2(\kappa)+S_2\big)^{-1}-S_1C_{22}(\kappa)\Big)S_2I_3(\kappa)^{-1}S_2\nonumber\\
&\qquad\times\Big(\big(I_2(\kappa)+S_2\big)^{-1}S_2+C_{22}(\kappa)S_1\Big)
\big(I_1(\kappa)+S_1\big)^{-1}S_0\big(I_0(\kappa)+S_0\big)^{-1}\nonumber\\
&=2\kappa\big(I_0(\kappa)+S_0\big)^{-1}\nonumber\\
&\quad+\Big(S_0\big(I_0(\kappa)+S_0\big)^{-1}-C_{00}(\kappa)\Big)S_0
\big(I_1(\kappa)+S_1\big)^{-1}S_0
\Big(\big(I_0(\kappa)+S_0\big)^{-1}S_0+C_{00}(\kappa)\Big)\nonumber\\
&\quad+\frac1\kappa\bigg\{\Big(S_1\big(I_0(\kappa)+S_0\big)^{-1}-C_{10}(\kappa)\Big)
\big(I_1(\kappa)+ S_1\big)^{-1}-\Big(S_0\big(I_0(\kappa)+S_0\big)^{-1}
-C_{00}(\kappa)\Big)C_{11}(\kappa)\bigg\}\nonumber\\
&\qquad\times S_1\big(I_2(\kappa)+S_2\big)^{-1}S_1\bigg\{\big(I_1(\kappa)+S_1\big)^{-1}
\Big(\big(I_0(\kappa)+S_0\big)^{-1}S_1+C_{10}(\kappa)\Big)\nonumber\\
&\qquad+C_{11}(\kappa)\Big(\big(I_0(\kappa)+S_0\big)^{-1}S_0
+C_{00}(\kappa)\Big)\bigg\}\nonumber\\
&\quad+\frac1{\kappa^2}\Bigg\{\bigg[\Big(S_2\big(I_0(\kappa)+S_0\big)^{-1}
-C_{20}(\kappa)\Big)\big(I_1(\kappa)+S_1\big)^{-1}\nonumber\\
&\qquad-\Big(S_0\big(I_0(\kappa)+S_0\big)^{-1}-C_{00}(\kappa)\Big)C_{21}(\kappa)\bigg]
\big(I_2(\kappa)+S_2\big)^{-1}\nonumber\\
&\qquad-\bigg[\Big(S_1\big(I_0(\kappa)+S_0\big)^{-1}-C_{10}(\kappa)\Big)
\big(I_1(\kappa)+S_1\big)^{-1}\nonumber\\
&\qquad-\Big(S_0\big(I_0(\kappa)+S_0\big)^{-1}
-C_{00}(\kappa)\Big)C_{11}(\kappa)\bigg]C_{22}(\kappa)\Bigg\}S_2I_3(\kappa)^{-1}S_2\nonumber\\
&\qquad\times\Bigg\{\big(I_2(\kappa)+S_2\big)^{-1}\bigg[\big(I_1(\kappa)+S_1\big)^{-1}
\Big(\big(I_0(\kappa)+S_0\big)^{-1}S_2+C_{20}(\kappa)\Big)\nonumber\\
&\qquad+C_{21}(\kappa)\Big(\big(I_0(\kappa)+S_0\big)^{-1}S_0+C_{00}(\kappa)\Big)\bigg]\nonumber\\
&\qquad+C_{22}(\kappa)\bigg[\big(I_1(\kappa)+S_1\big)^{-1}
\Big(\big(I_0(\kappa)+S_0\big)^{-1}S_1+C_{10}(\kappa)\Big)\nonumber\\
&\qquad+C_{11}(\kappa)\Big(\big(I_0(\kappa)+S_0\big)^{-1}S_0
+C_{00}(\kappa)\Big)\bigg]\Bigg\}.\label{grosse_formule}
\end{align}
The interest in this formulation is that the projections $S_j$ (which lead to
simplifications in the proof) have been put into evidence at the beginning or at the
end of each term.

\begin{proof}
(a) Some lengthy, but direct, computations taking into account the expansion
\eqref{grosse_formule}, the relation $\big(I_j(0)+S_j\big)^{-1}S_j=S_j$, the expansion
\eqref{dev1} for $\F_0(\lambda-\kappa^2;n,\sigma)$ and
$\F_0(\lambda-\kappa^2;n',\sigma')$ and Lemma \ref{help1}(a) lead to the equality
\begin{align*}
&\lim_{\kappa\to0}\F_0(\lambda-\kappa^2;n,\sigma)v\;\!\M(\lambda,\kappa)
v\F_0(\lambda-\kappa^2;n',\sigma')^*\\
&=\F_0(\lambda;n,\sigma)\;\!vS_0\big(I_1(0)+S_1\big)^{-1}S_0v\F_0(\lambda;n',\sigma')^*\\
& \quad-\F_0(\lambda;n,\sigma)\;\!v
\big(C_{20}'(0)+S_0 C_{21}'(0)\big)S_2I_3(0)^{-1}S_2
\big(C_{20}'(0)+C_{21}'(0)S_0\big)v\F_0(\lambda;n',\sigma')^*.
\end{align*}
Since
\begin{equation}\label{start}
S(\lambda-\kappa^2 ;n,\sigma,n',\sigma')-\delta_{n\sigma n'\sigma'}
=-2\pi i\F_0(\lambda-\kappa^2;n,\sigma)v\;\!\M(\lambda,\kappa)v
\F_0(\lambda-\kappa^2;n',\sigma')^*,
\end{equation}
this proves the claim.

(b.1) We first consider the case $\lambda_n<\lambda$, $\lambda_{n'}=\lambda$ (the case
$\lambda_n=\lambda$, $\lambda_{n'}<\lambda$ is not presented since it is similar). An
inspection taking into account the expansion \eqref{grosse_formule}, the relation
$\big(I_j(\kappa)+S_j\big)^{-1}=\big(I_j(0)+S_j\big)^{-1}+\Oa(\kappa)$ and the relation
$\big(I_j(0)+S_j\big)^{-1}S_j=S_j$ leads to the equation
\begin{align}
&\F_0(\lambda-\kappa^2;n,\sigma)\;\!v\;\!\M(\lambda,\kappa)\;\!v\;\!
\F_0(\lambda-\kappa^2;n',\sigma')^*\nonumber\\
&=\F_0(\lambda-\kappa^2;n,\sigma)\;\!v\;\!
\bigg\{\Oa(\kappa)+S_0\big(I_1(\kappa)+S_1\big)^{-1}S_0\nonumber\\
&\quad+\frac1\kappa\big(S_1+\Oa(\kappa)\big)S_1
\big(I_2(\kappa)+S_2\big)^{-1}S_1\big(S_1+\Oa(\kappa)\big)\nonumber\\
&\quad+\frac1{\kappa^2}\Big[\Oa(\kappa^2)+S_2\big(I_0(\kappa)+ S_0\big)^{-1}
\big(I_1(\kappa)+S_1\big)^{-1}\big(I_2(\kappa)+S_2\big)^{-1}
-C_{20}(\kappa)-S_0C_{21}(\kappa)\nonumber\\
&\qquad-S_1C_{22}(\kappa)\Big]S_2 I_3(\kappa)^{-1}S_2
\Big[\Oa(\kappa^2)+\big(I_2(\kappa)+S_2\big)^{-1}\big(I_1(\kappa)+S_1\big)^{-1}
\big(I_0(\kappa)+S_0\big)^{-1}S_2\nonumber\\
&\qquad+C_{20}(\kappa)+C_{21}(\kappa)S_0+C_{22}(\kappa)S_1\Big]\bigg\}
\;\!v\;\!\F_0(\lambda-\kappa^2;n',\sigma')^*.\label{eq_bordeaux}
\end{align}
Applying Lemma \ref{help1} to the previous equation gives
\begin{align*}
&\F_0(\lambda-\kappa^2;n,\sigma)\;\!v\;\!\M(\lambda,\kappa)\;\!v\;\!
\F_0(\lambda-\kappa^2;n',\sigma')^*\\
&=\F_0(\lambda-\kappa^2;n,\sigma)\;\!v\;\!
\bigg\{\Oa(\kappa)-\frac1{\kappa^2}\big(\O(\kappa^2)+C_{20}(\kappa)+S_0C_{21}(\kappa)\big)
S_2 I_3(\kappa)^{-1}S_2\big(\Oa(\kappa^2)+C_{20}(\kappa)\big)\bigg\}\\
&\qquad\times v\;\!\F_0(\lambda-\kappa^2;n',\sigma')^*.
\end{align*}
Finally, taking into account the expansion \eqref{dev1} for
$\F_0(\lambda-\kappa^2;n,\sigma)$ and the expansion \eqref{dev2} for
$\F_0(\lambda-\kappa^2;n',\sigma')$, one ends up with
\begin{align}
&\F_0(\lambda-\kappa^2;n,\sigma)\;\!v\;\!\M(\lambda,\kappa)\;\!v\;\!
\F_0(\lambda-\kappa^2;n',\sigma')^*\nonumber\\
&=(-\kappa^2)^{-5/4}\F_0(\lambda;n,\sigma)\;\!v\;\!\big(\O(\kappa^2)+C_{20}(\kappa)
+S_0C_{21}(\kappa)\big)S_2I_3(\kappa)^{-1}S_2\big(\Oa(\kappa^2)+C_{20}(\kappa)\big)
\;\!v\;\!\gamma_0(n')^*\nonumber\\
&\quad+\O(|\kappa|^{1/2})\label{fin1},
\end{align}
where $\gamma_0(n')^*$ is given by
$\gamma_0(n')^*\psi=\frac1{2\sqrt\pi}\;\!\psi\otimes1$ for any
$\psi\in\P_{n'}\;\!\ltwo(\Sigma)$.

Now, Lemma \ref{la_cle_des_champs}(c) implies that $[M_1(0),S_2]=0$, and thus that
\begin{equation}\label{eq_C20}
C_{20}(\kappa)
=2\kappa\;\!\big(I_0(0)+S_0\big)^{-1}[M_1(0),S_2]\big(I_0(0)+S_0\big)^{-1}
+\O(\kappa^2)\\
=\O(\kappa^2).
\end{equation}
In consequence, one infers from \eqref{fin1} that
$
\F_0(\lambda-\kappa^2;n,\sigma)\;\!v\;\!\M(\lambda,\kappa)\;\!v\;\!
\F_0(\lambda-\kappa^2;n',\sigma')^*
$
vanishes as $\kappa\to0$, and thus that the limit
$\lim_{\kappa\to0}S(\lambda-\kappa^2;n,\sigma,n',\sigma')$ also vanishes by
\eqref{start}.

(b.2) We are left with the case $\lambda_n=\lambda=\lambda_{n'}$. An inspection of the
expansion \eqref{grosse_formule} taking into account the relation
$\big(I_\ell(\kappa)+S_\ell\big)^{-1}=\big(I_\ell(0)+S_\ell\big)^{-1}+\Oa(\kappa)$,
the relation $\big(I_\ell(0)+S_\ell\big)^{-1}S_\ell=S_\ell$ and Lemma \ref{help1}(b)
leads to the equation
\begin{align*}
&\F_0(\lambda-\kappa^2;n,\sigma)\;\!v\;\!\M(\lambda,\kappa)\;\!v\;\!
\F_0(\lambda-\kappa^2;n',\sigma')^*\\
&=\F_0(\lambda-\kappa^2;n,\sigma)\;\!v\;\!\bigg\{\Oa(\kappa^2)
+\kappa\big(I_0(\kappa)+S_0\big)^{-1}
-\frac1\kappa\;\!C_{10}(\kappa)S_1\big(I_2(\kappa)+S_2\big)^{-1}S_1C_{10}(\kappa)\\
&\quad-\frac1{\kappa^2}\;\!\big(\Oa(\kappa^2)+C_{20}(\kappa)\big)S_2 I_3(\kappa)^{-1}S_2
\big(\Oa(\kappa^2)+C_{20}(\kappa)\big)\bigg\}\;\!v\;\!\F_0(\lambda-\kappa^2;n',\sigma')^*.
\end{align*}
Therefore, the expansion \eqref{dev2} for $\F_0(\lambda-\kappa^2;n,\sigma)$ and
$\F_0(\lambda-\kappa^2;n',\sigma')$ and the inclusion $C_{20}(\kappa)\in\O(\kappa^2)$
(see \eqref{eq_C20}), imply that the limit
$$
\lim_{\kappa\to0}\F_0(\lambda-\kappa^2;n,\sigma)
\;\!v\;\!\M(\lambda,\kappa)\;\!v\;\!\F_0(\lambda-\kappa^2;n',\sigma')^*
$$
exists, and thus that the limit
$\lim_{\kappa\to0}S(\lambda-\kappa^2;n,\sigma,n',\sigma')$ also exists by \eqref{start}.
\end{proof}

We finally consider the continuity of the scattering matrix at embedded eigenvalues
not located at thresholds.

\begin{Proposition}
Suppose that $V\in\linf(\Omega;\R)$ has bounded support and take
$\lambda\in\sigma_{\rm p}(H)\setminus\tau$, $\kappa\in\partial O(\varepsilon)$ with
$\varepsilon>0$ small enough, $n,n'\ge1$, and $\sigma,\sigma'\in\{+,-\}$. Then, if
$\lambda_n<\lambda$ and $\lambda_{n'}<\lambda$, the limit
$\lim_{\kappa\to0}S(\lambda-\kappa^2;n,\sigma,n',\sigma')$ exists.
\end{Proposition}

\begin{proof}
We know from \eqref{eq_expansion_2} that
$$
\M(\lambda,\kappa)
=\big(J_0(\kappa)+S\big)^{-1}+\frac1{\kappa^2}\;\!\big(J_0(\kappa)+S)^{-1}S
J_1(\kappa)^{-1}S\big(J_0(\kappa)+S\big)^{-1},
$$
with $S$ the Riesz projection associated with the value $0$ of the operator
$
T_0=u+\sum_nv\big(\P_n\otimes R^0(\lambda-\lambda_n)\big)\;\!v.
$
Now, a commutation of $S$ with $\big(J_0(\kappa)+S\big)^{-1}$ gives
$$
\M(\lambda,\kappa)
=\big(J_0(\kappa)+S\big)^{-1}
+\frac1{\kappa^2}\;\!\big\{S\big(J_0(\kappa)+S)^{-1}+\Oa(\kappa)\big\}
SJ_1(\kappa)^{-1}S\;\!\big\{\big(J_0(\kappa)+S\big)^{-1}S+\Oa(\kappa)\big\},
$$
and a computation as in the proof of Lemma \ref{help1}(a) (but which takes directly
Lemma \ref{Cor_magique} into account) shows that
$\F_0(\lambda-\kappa^2;n,\sigma)\;\!vS\in\O(\kappa^2)$ and
$Sv\;\!\F_0(\lambda-\kappa^2;n',\sigma')^* \in \O(\kappa^2)$. These estimates,
together with the expansion \eqref{dev1} for $\F_0(\lambda-\kappa^2;n,\sigma)$ and
$\F_0(\lambda-\kappa^2;n',\sigma')^*$ and the equation \eqref{start}, imply the claim.
\end{proof}



\end{document}